\documentclass[12pt]{article}

\usepackage{hyperref}
\hypersetup{colorlinks=true}
\usepackage[sort&compress,numbers]{natbib}
\usepackage{doi}
\usepackage{enumitem}
\setlist{itemsep=1pt}

\usepackage{amsmath, amsfonts, amsthm}
\usepackage{bm}
\usepackage{mathtools}
\usepackage{xcolor}
\newtheorem{lemma}{Lemma}
\newtheorem{prop}{Proposition}
\newtheorem{coro}{Corollary}
\newtheorem{theorem}{Theorem}
\theoremstyle{remark}

\newtheorem{example}{Example}

\def\Me{\mathcal M}
\def\Fe{\mathcal F}
\def\Ne{\mathcal N}

\def\Pe{\mathcal P}

\def\Se{\mathcal S}

\def\Tr{\mathrm{Tr}\,}
\def\<{\langle\,}
\def\>{\,\rangle}

\def\eexp{E_{\exp}}
\def\llog{L_{\log}}

\title{The exponential Orlicz space in quantum information geometry}

\author{Anna Jen\v cov\' a\\
Mathematical Institute, Slovak Academy of Sciences}

\begin{document}
\maketitle

\abstract{We review the construction of a quantum version of the exponential statistical
manifold over the set of all faithful normal positive functionals on a von Neumann
algebra. The construction is based on the relative entropy approach to state perturbation.
We construct a quantum version of the exponential Orlicz space  and discuss
the properties of this space and its dual with respect to  Kosaki $L_p$-spaces. We show
that the constructed manifold admits a canonical divergence satisfying a Pythagorean
relation. We also prove that the manifold structure is invariant under  sufficient
channels.}

\maketitle

\section{Introduction} One of the fundamental achievements of Information geometry is
the rigorous extension from parametric statistical models to the nonparametric case by
Pistone and Sempi, \cite{pistone1995aninfinite}, who constructed a Banach manifold structure on the set of 
probability measures equivalent to a given probability measure. The manifold  structure is
based on an Orlicz space associated to an exponential Young function
$\Phi(x)=\cosh(x)-1$. This theory has been subsequently developed in a number of works,
see e.g. 
\cite{pistone1999theexponential, cena2007exponential}. 
In this construction,  the properties of the moment generating function and its conjugate,  the Kullback Leibler divergence
(relative entropy), play a central role.

To obtain a full quantum version of the Pistone-Sempi construction
would mean to introduce an analogous  Banach manifold structure on the set of faithful normal states of a
general ($\sigma$-finite) von Neumann algebra. The problem is that known versions
of quantum Orlicz spaces are either restricted to the semifinite case
(e.g. \cite{kunze1990noncommutative, alrashed2007noncommutative}) or are
technically quite involved (\cite{labuschagne2013acrossed}) and it
is unclear how to introduce an  exponential structure on the set of states, based on these
spaces. 

Another approach using perturbation of states on the algebra $B(\mathcal H)$ of bounded
operators on a Hilbert space $\mathcal H$, was taken in \cite{grasselli2000thequantum, streater2004quantum}. Here
the manifold is modelled on the space of certain (unbounded) perturbation operators, which
is given the structure of a Banach space. In \cite{streater2004quantum}, the Banach space  is constructed  from the free energy functional, which can be seen as the
counterpart of the classical cumulant generating function. 

This idea inspired the works \cite{jencova2006aconstruction,jencova2010onquantum}, where a
definition of the exponential Orlicz space over a von Neumann algebra $\Me$ with respect
to a faithful normal state $\rho$ is proposed, using the relative entropy approach
to state perturbation. This definition has the advantage that it is based directly on the
relative entropy $S(\cdot\|\cdot)$ and its properties. In particular, the set $\Pe_\rho$ of all normal positive linear
functionals such that $S(\omega\|\rho)<\infty$ is identified with a generating cone in the
dual of the constructed exponential Orlicz space, so that a manifold structure on normal states
of $\Me$, respecting the
relative entropy,  can be introduced by immersion into a Banach space. Moreover, an
exponential 
manifold structure is obtained using perturbations of the state $\rho$ by elements of the
exponential Orlicz space and the connected components of this manifold are contained in
$\Pe_\rho$.

In the present paper, we review the construction of the exponential Orlicz space and its
dual, as defined in \cite{jencova2006aconstruction, jencova2010onquantum}.  We present the proofs in a more streamlined and precise form. 
The dual space is found 
explicitly as an Orlicz space, using the conjugate Young function. We show the relation 
of the constructed spaces to the Kosaki $L_p$-spaces. The manifold structure is 
introduced over the
positive cone of faithful positive linear functionals, rather than states,  similarly 
to the approach in \cite{ay2017information}. We  define a canonical divergence on the manifold,
satisfying a generalized Pythagorean relation. Finally, we prove the invariance of our
structures under sufficient channels, which is the counterpart of the important invariance property of 
the classical information geometry.

\section{The exponential Orlicz space}

In this section, we review the definition of the exponential Orlicz space from
\cite{jencova2006aconstruction}, construct its dual as an Orlicz space and study some of
the properties of these spaces.

\subsection{A general construction of an Orlicz space}\label{sec:orlica}

Let $X$ be a real vector space. A function $\Phi:X\to [0,\infty]$ is called a Young function if
it satisfies:
\begin{enumerate}
\item[(i)] $\Phi$ is convex,
\item[(ii)] $\Phi(x)=\Phi(-x)$ for all $x\in X$ and $\Phi(0)=0$,
\item[(iii)] if $x\ne 0$ then $\lim_{t\to \infty} \Phi(tx)=\infty$.

\end{enumerate}
For a Young function $\Phi$, put  $C_\Phi:=\{x\in X, \Phi(x)\le 1\}$ and $V_\Phi:=\{x\in \Phi(x),\ \exists s>0, \
\Phi(sx)<\infty\}$.
The set  $C_\Phi$ is absolutely convex and $V_\Phi=\bigcup_n nC_\Phi$ is the linear span of the
effective domain $\mathrm{Dom}(\Phi)=\{x\in X,\ \Phi(x)<\infty\}$. We
can define a norm in $V_\Phi$ as the Minkowski functional of $C_\Phi$:
\[
\|x\|_\Phi:=\inf\{\lambda>0,\ \Phi(\frac{x}{\lambda})\le 1\},\qquad x\in V_\Phi.
\]
The completion of $V_\Phi$ with respect to this norm will be denoted by $B_\Phi$. 

Let  $(\Omega,\Sigma,\mu)$ be a measure space and  let $X$  be the vector space 
 of measurable functions $f:\Omega\to \mathbb R$. Let
$\varphi$ be a Young function on $\mathbb R$ and put
\[
\Phi_\varphi(f):=\int_\Omega\varphi(\lvert f\rvert)d\mu.
\]
Then $B_{\Phi_\varphi}$ is the classical Orlicz space $L^\varphi(\Omega,\Sigma,\mu)$ and
$\|\cdot\|_{\Phi_\varphi}$ is the Luxemburg-Nakano norm, \cite{rao1991theory}.   As
another example, let $\Me$ be a  semifinite von Neumann algebra
 with a  faithful normal semifinite trace $\tau$. Let  $X$ be the vector
space of self-adjoint $\tau$-measurable operators and
$\Phi_\varphi(x)=\tau(\varphi(\lvert x \rvert))$, then we obtain the Orlicz space
$L^\varphi(\Me,\tau)$, \cite{kunze1990noncommutative}. In
\cite{alrashed2007noncommutative}, a version of the Orlicz space with
respect to a faithful normal state $\rho$ was also defined  by a construction of this
type. See also
\cite{labuschagne2013acrossed} in the case of a general von Neumann algebra with a faithful normal
weight.

It was shown in \cite[Prop. 2]{jencova2006aconstruction} that if $X$ is a Banach space and  $\Phi$
is continuous, then $X$ is continuously embedded in $B_\Phi$. The conjugate function
$\Phi^*$ is again a Young function such that
 $V_{\Phi^*}=\mathrm{Dom}(\Phi^*)=B_{\Phi^*}$ and we have a continuous embedding
 $B_{\Phi^*}\sqsubseteq X^*$. Moreover,  $B_{\Phi^*}=B_{\Phi}^*$, with equivalent norms.

\subsection{Basic setting and notations}

We briefly describe the setting of von Neumann algebras and noncommutative $L_p$-spaces.
For a quick overview of these topics, see \cite{hiai2021lectures}.

Let  $\Me$ be a ($\sigma$-finite) von Neumann algebra. We will denote by $\Me^*$ the dual
space of $\Me$ and by $\Me_*$ the predual, consisting of normal functionals in $\Me^*$.
The positive cones in these spaces will be denoted by $\Me^+$, $(\Me^*)^+$ and $\Me_*^+$.
An element $\rho\in \Me_*^+$ is faithful if $\rho(a)=0$ implies $a=0$, for any $a\in \Me^+$.

For $1\le p\le \infty$, we
denote the Haagerup $L_p$-space over $\Me$  by $L_p(\Me)$ and its norm  by  
$\|\cdot\|_p$. 
We will use the identification of $L_\infty(\Me)$ with $\Me$ and  $L_1(\Me)$ with $\Me_*$.
Let $h_\psi\in L_1(\Me)$ be the element corresponding to  $\psi\in \Me_*$, then we can
define the trace in $L_1(\Me)$ by  $\Tr[ h_\psi]=\psi(1)$. 

For $p,q,r\ge 1$ such that $1/p+1/q=1/r$ and $h\in L_p(\Me)$, $k\in L_q(\Me)$, we have
$hk\in L_r(\Me)$ and the H\"older inequality holds:
\[
\|hk\|_r\le \|h\|_p\|k\|_q.
\]
For $1\le p<\infty$ and $1/p+1/q=1$, the space $L_q(\Me)$ can be identified with the dual
space $L_p(\Me)^*$, with duality given by
\[
\<h,k\>=\Tr[hk],\qquad h\in L_p(\Me),\ k\in L_q(\Me).
\]
The space $L_2(\Me)$ is a Hilbert space with inner product 
\[
(h,k)= \Tr [h^*k],\qquad h,k\in L_2(\Me).
\]
We will use the representation of $\Me$ on $L_2(\Me)$ by the  left action $\lambda(a): h\mapsto ah$ for $a\in \Me$ and $h\in
L_2(\Me)$. The quadruple $(\lambda(\Me), L_2(\Me), L_2(\Me)^+,J)$, where $L_2(\Me)^+$ is
the cone of positive operators in $L_2(\Me)$ and $J$ is defined by $Jh=h^*$,  is a standard
form of $\Me$ (\cite[Thm.
3.6]{terp1981lpspaces}, \cite[Thm. 9.29]{hiai2021lectures}).   For more on the standard form see
\cite{takesaki2003theory2} or \cite[Sec. 3]{hiai2021lectures}. Any positive normal functional $\varphi\in \Me_*^+$ has a
unique vector representative $h_\varphi^{1/2}$ in the cone $L_2(\Me)^+$, that is,
\[
\varphi(a)=(h_\varphi^{1/2},ah_\varphi^{1/2}),\qquad a\in \Me.
\]

Let us now fix a faithful positive normal functional $\rho\in \Me_*^+$. The (symmetric) Kosaki
$L_p$-space with respect to $\rho$ \cite{kosaki1984applications} is defined via complex
interpolation, using the continuous embedding
\begin{equation}\label{eq:infty_embedding}
i_{\infty,\rho}: \Me\to L_1(\Me),\qquad a\mapsto h_\rho^{1/2}a h_\rho^{1/2}.
\end{equation}
Let us denote the image $i_{\infty,\rho}(\Me)$ by $L_\infty(\Me,\rho)$, with the norm
$\|i_{\infty,\rho}(a)\|_{\infty,\rho}=\|a\|$. 
The  interpolation space $C_{1/p}(L_\infty(\Me,\rho),L_1(\Me))$
\cite{bergh1976interpolation} will be denoted by
$L_p(\Me,\rho)$ and the norm by $\|\cdot\|_{p,\rho}$. The map
\begin{equation}\label{eq:pembedding}
i_{p,\rho}: L_p(\Me)\to L_1(\Me),\qquad k\mapsto h_\rho^{1/2q}kh_\rho^{1/2q}
\end{equation}
with $1/p+1/q=1$ is an isometric isomorphism of $L_p(\Me)$ onto $L_p(\Me, \rho)$ for
$1\le p\le \infty$. From the properties of complex interpolation spaces, we have for
$1\le p'\le p\le\infty$ the 
continuous embeddings $L_p(\Me,\rho)\sqsubseteq L_{p'}(\Me,\rho)\sqsubseteq L_1(\Me)$. We have $L_q(\Me,\rho)\simeq L_p(\Me,\rho)^*$ for $1\le p<\infty$
and  $1/p+1/q=1$,  with duality given by
\[
\<i_{p,\rho}(k),i_{q,\rho}(l)\>=\Tr[kl],\qquad k\in L_p(\Me),\ l\in L_q(\Me).
\]
Note also that the Kosaki $L_p$-spaces can be constructed as in Section \ref{sec:orlica},
where  $X=\Me^s:=\{a=a^*\in \Me\}$  and $\Phi(a)=\|h_\rho^{1/2p}ah_\rho^{1/2p}\|_p$,
\cite{zolotarev1982lpspaces}.

Let $\Ne$ be another  von Neumann algebra and let $T:L_1(\Me)\to L_1(\Ne)$ be a positive linear
map that preserves trace. Such a map will be called a channel. The adjoint of $T$ is a
positive unital  normal map $T^*:\Ne \to \Me$. 

Let $\rho$ be a faithful element in $\Me^+_*$. One can see
(\cite[Sec. 3.3]{jencova2018renyi}) that if $e=s(T(\rho))$ is the support projection  of
$T(\rho)$, then we have $T(\omega)=eT(\omega)e$ for all $\omega\in L_1(\Me)$, hence we may suppose that $T(\rho)$ is
faithful by replacing $\Ne$ by $e\Ne e$.

\begin{prop}\label{prop:lp_contraction}\cite{jencova2018renyi}
The restriction of a channel $T$ to $L_p(\Me,\rho)$ is  contraction $L_p(\Me,\rho)\to
L_p(\Ne, T(\rho))$, for any
$1\le p\le \infty$.

\end{prop}

In the case $p=\infty$,  there is a positive linear map $T^*_\rho:
\Me\to \Ne$, defined by
\[
T(h_\rho^{1/2}ah_\rho^{1/2})=T(\rho)^{1/2}T^*_\rho(a) T(\rho)^{1/2},\qquad a\in \Me.
\]
 The map $T^*_\rho$ was introduced in \cite{petz1988sufficiency}
and is called the Petz dual of $T$ (with
respect to $\rho$). It was also proved that   $T_\rho^*$ is unital and normal, moreover, it is $n$-positive if and only
if $T$ is $n$-positive, for any $n$. Let $T_\rho: L_1(\Ne)\to L_1(\Me)$ be the preadjoint
of $T^*_\rho$.  Then  $T_\rho\circ T(\rho)=\rho$ and the Petz
dual of $T_\rho$ is $T^*$.

\subsection{Relative entropy and related functionals}\label{sec:relent}

The Araki relative entropy for $\omega,\rho\in \Me_*^+$ \cite{araki1976relative,ohya1993quantum} is defined using the
relative modular operator $\Delta_{\rho,\omega}(=\Delta_{\rho,h_\omega^{1/2}})$ as
\[
S(\omega\|\rho)= \begin{dcases} -\<\log(\Delta_{\rho,\omega})h_\omega^{1/2},h^{1/2}_\omega\> &
\text{if } s(\omega)\le s(\rho)\\
\infty& \text{otherwise}.
\end{dcases}
\]
Here  $s(\rho)$ denotes the support projection of $\rho$. Alternatively, we have the 
following variational formula due to Kosaki \cite{kosaki1986relative}:
\begin{equation}\label{eq:kosaki}
S(\omega\|\rho)=\sup_n \sup \left\{\omega(1)\log
n-\int_{1/n}^\infty(\omega(y(t)^*y(t))+t^{-1}\rho(x(t)x(t)^*)\frac{dt}{t}\right\}
\end{equation}
here the second supremum is taken over all step functions $x:(1/n,\infty)\to L$ with
finite range, $y(t)=1-x(t)$ and $L$ is a subspace in $\Me$ containing 1 which is dense in
the strong*-operator topology. 

The relative entropy $S$ is a jointly convex function 
$S:\Me_*^+\times \Me_*^+\to \mathbb R\cup\{\infty\}$,  lower semicontinuous with respect to the
$\sigma(\Me_*,\Me)$-topology. Moreover, $S$ is strictly convex in the first variable,
which can be inferred from the identity \cite[Prop. 5.22]{ohya1993quantum} 
\begin{equation}\label{eq:donald}
S(\omega\|\rho)+\sum_i S(\omega_i\|\omega)=\sum_i S(\omega_i\|\rho),\quad
\omega=\sum_{i=1}^k\omega_i,\ \omega\in \Me_*^+.
\end{equation}
Note also that since $\omega_i\le \omega$ in \eqref{eq:donald}, we have
$S(\omega_i\|\omega)\le S(\omega_i\|\omega_i)<\infty$ for all $i$.
See  \cite[Sec. 5]{ohya1993quantum} for details and a list of further important properties of $S$. 
The next statement shows the relation to the Kosaki $L_p$-space.

\begin{prop}\cite{jencova2018renyi, hiai2021quantum}\label{prop:relent_lp} Let $\omega,\rho\in \Me_*^+$ be such that $h_\omega\in L_p(\Me,\rho)$ for
some $p>1$. Then the function $f:(1,p]\to \mathbb R$, defined as 
\[
f(\alpha):=
\frac1{\alpha-1}\log\frac{\|h_\omega\|_{\alpha,\rho}^\alpha}{\omega(1)}
\]
is increasing and $\lim_{\alpha\downarrow
1}f(\alpha)=\frac1{\omega(1)}S(\omega\|\rho)$.
\end{prop}

Using the Kosaki variational formula \eqref{eq:kosaki}, the relative entropy can be readily
extended to  a function $S:(\Me^*)^+\times
(\Me^*)^+\to \mathbb R\cup\{\infty\}$. 

\begin{prop}\label{prop:relentropy_extension} Let $\rho\in \Me_*^+$ and let $\omega\in
(\Me^*)^+$. If $\omega\notin \Me_*^+$, then $S(\omega\|\rho)=\infty$.

\end{prop}

\begin{proof} \cite{ohya1993quantum} There is another way to define the relative entropy
for elements in $(\Me^*)^+$. Let $(\pi_u,\mathcal H_u)$ be the universal representation of $\Me$ and let
$\bar \Me=\pi_u(\Me)''\cong \Me^{**}$ be the universal enveloping von Neumann algebra of
$\Me$
\cite{takesaki2002theory1}. Then each element of the  dual space $\omega\in \Me^*$ has a
unique extension to a normal functional 
$\bar\omega$ on $\bar \Me$ and $\Me^*$  is the predual of $\bar \Me$. Moreover, there is a central
projection $z_0\in \bar \Me$ such that $\Me_*=\Me^*z_0$. We can define for
$\omega,\rho\in (\Me^*)^+$ the relative entropy $\bar S: (\Me^*)^+\times (\Me^*)^+\to
\mathbb R$ as
\[
\bar S(\omega\|\rho):=S_{\bar \Me}(\bar\omega\|\bar\rho)
\]
(here $S_{\bar \Me}$ is computed with respect to the von Neumann algebra $\bar \Me$).
Now note that we may use $L=\pi_u(\Me)$ in the variational formula  for
$S_{\bar\Me}$ and $L=\Me$ for $S$, which implies that $\bar S=S$. Let $\rho\in \Me_*^+$,
$\omega\in (\Me^*)^+$ and 
assume that  $\omega$ is not normal. Then we must have $\bar \omega(1-z_0)>0$ but 
$\bar \rho(1-z_0)=0$, so that $s(\bar \omega)\not\le s(\bar \rho)$. By definition of the relative entropy, this
implies that $\bar S(\omega\|\rho)=\infty$.  

\end{proof}

From now on, let us fix a faithful normal functional  $\rho\in \Me_*^+$. Let $\Me^s$ denote the real vector subspace of self-adjoint elements of $\Me$. Then
$\Me^s$ is closed in $\Me$ and its Banach space dual is the space $(\Me^*)^s$ of all linear
functionals  $\varphi\in \Me^*$ satisfying $\varphi(a^*)=\overline{\varphi(a)}$, $a\in
\Me$. Note that we have $(\Me^*)^s=(\Me^*)^+-(\Me^*)^+$. Similarly, $\Me_*^s=(\Me^*)^s\cap \Me_*$ is the predual of $\Me^s$
 and $\Me_*^s=\Me_*^+-\Me_*^+$.
Let us define
the function $F_\rho:(\Me^*)^s\to \mathbb R$ by
\[
F_\rho(\omega):=\begin{dcases} S(\omega\|\rho)-\omega(1) &\text{if } \omega\in (\Me^*)^+\\
\infty & \text{otherwise}.
\end{dcases}
\]
We also define the sets
\[
\Se_C:=\{\omega\in (\Me^*)^s,\ F_\rho(\omega)\le C\},\ C\in \mathbb R, \qquad \mathcal
P_\rho:=\{\omega\in (\Me^*)^s,\ F_\rho(\omega)<\infty\}.
\]
In other words,  $\Pe_\rho$ is the effective domain of $F_\rho$.  Note that we  have $\Se_C\subseteq \Pe_\rho\subseteq \Me_*^+$, by Proposition
\ref{prop:relentropy_extension}. The next proposition lists some important properties of
the function 
$F_\rho$ and these sets.

\begin{prop}\label{prop:Frho_properties}
\begin{enumerate}
\item[(i)] $F_\rho:(\Me^*)^s\to \mathbb R$ is strictly convex and lower semicontinuous in
the $\sigma((\Me^*)^s,\Me^s)$ topology.
\item[(ii)] We have the inequalities
\[
F_\rho(\omega)\ge \omega(1)(\log\frac{\omega(1)}{\rho(1)}-1)\ge -\rho(1).
\]
The first inequality becomes an equality if and only if $\omega=\lambda\rho$ for some
$\lambda\ge 0$. In particular, $F_\rho(\omega)=-\rho(1)$ if and only if $\omega=\rho$.

\item[(iii)] For any $C\in \mathbb R$, $\mathcal S_C$
is  convex and  compact in both the $\sigma((\Me^*)^s,\Me^s)$  and the
$\sigma(\Me^s_*,\Me^s)$-topology.
\item[(iv)] The set $\mathcal P_\rho$ 
is a face of the cone $\Me_*^+$, containing $L_p(\Me,\rho)^+$ for any $1<p\le \infty$.

\end{enumerate}

\end{prop}

\begin{proof} The proof of (i)-(ii) follows from the variational formula and properties
of $S$.
For the proof of (iii), let $\omega\in \Se_C$, then by (ii),
\[
\omega(1)(\log\frac{\omega(1)}{\rho(1)}-1)\le F_\rho(\omega)\le C.
\]
This implies that $\omega(1)=\|\omega\|$ must be bounded over $\Se_C$.
Since $\sigma((\Me^*)^s,\Me^s)$ is the weak*-topology on $(\Me^*)^s$ and $\Se_C$ is closed by (i), this implies  that
$\Se_C$ is compact. But   $\Se_C\subseteq \Me^s_*$, so that  it is also compact in the
$\sigma(\Me^s_*,\Me^s)$-topology. 

To prove the last statement (iv), let $\omega=\sum_i\omega_i$ for some $\omega_i\in
\Me_*^+$. Then by \eqref{eq:donald}
\[
F_\rho(\omega)+\sum_i S(\omega_i\|\omega)=\sum_i F_\rho(\omega_i).
\]
Since $S(\omega_i\|\omega)<\infty$, $\omega\in \mathcal P_\rho$ if and only if all
$\omega_i\in \mathcal P_\rho$, so that $\Pe_\rho$ is a face of $\Me_*^+$. The fact that
 $L_p(\Me,\rho)^+\subseteq \mathcal P_\rho$ for $1<p\le \infty$ follows from Proposition
 \ref{prop:relent_lp}.

\end{proof}

We also have the following important monotonicity property.

\begin{prop}\label{prop:monotone_F} Let $T:L_1(\Me)\to L_1(\Ne)$ be a channel. Then
\[
F_{T(\rho)}(T(\omega))\le F_\rho(\omega),\qquad \omega \in (\Me^*)^s. 
\]

\end{prop}

\begin{proof} The statement follows from \cite[Sec. 3.1 and Thm. 4.1]{jencova2021renyi}.

\end{proof}

We  next study the Legendre-Fenchel conjugate of $F_\rho$
with respect to the dual pair $((\Me^*)^s,\Me^s)$, see e.g. \cite{ekeland1999convex} or
\cite{zalinescu2002convex} for more information on Legendre-Fenchel duality of convex
functions. Namely, we define the function $C_\rho$ on $\Me^s$ as
\begin{equation}\label{eq:frho_conjugate}
C_\rho(a):= F_\rho^*(a)=\sup_{\omega\in (\Me^*)^s} \omega(a)-F_\rho(\omega), \qquad a\in
\Me^s.
\end{equation}
The proof of the following result can be obtained from \cite[Sec.
12]{ohya1993quantum}. See also \cite{petz1988variational, donald1990relative}. We collect the
arguments for convenience of the reader.

\begin{theorem}\label{thm:conjugate}  The supremum in \eqref{eq:frho_conjugate} is attained
at a unique functional  $\rho^a\in \Me_*^+$. The element $\rho^a$ is faithful  and  $C_\rho(a)=\rho^a(1)$.
Moreover, we have the equality
\begin{equation}\label{eq:perturbed_entropy}
\omega(a)+S(\omega\|\rho^a)=S(\omega\|\rho),\qquad \omega\in \Me_*^+
\end{equation}
and the chain rule 
\begin{equation}\label{eq:chin_rule}
\rho^{a+b}=(\rho^a)^b,\quad C_\rho(a+b)=C_{\rho^a}(b),\qquad a,b\in \Me^s.
\end{equation}

\end{theorem}

\begin{proof} Let $a\in \Me^s$ and let $\xi(a)$ denote the perturbed vector
\cite{araki1973relative}
\[
\xi(a)=\sum_{n=0}^\infty
\int_0^{1/2}dt_1\int_0^{t_1}dt_2\dots\int_0^{t_n}dt_n\Delta_\rho^{t_n}a\Delta_\rho^{t_{n-1}-t_n}a\dots\Delta_\rho^{t_1-t_2}ah_\rho^{1/2}.
\]
Then $\xi(a)\in L_2(\Me)^+$ and the functional $\rho^a\in \Me_*^+$ given by
$(\xi(a),\cdot\xi(a))$ is faithful.  By \cite[Thm. 3.10]{araki1977relative}, $\rho^a$
satisfies \eqref{eq:perturbed_entropy}. It follows that for $\omega\in \Pe_\rho$,
\[
\omega(a)-F_\rho(\omega)=-F_{\rho^a}(\omega)\le \rho^a(1),
\]
with equality if and only if $\omega=\rho^a$ (Proposition \ref{prop:Frho_properties}
(iii)). By replacing $\rho$ by $\rho^b$ in
\eqref{eq:perturbed_entropy}, we obtain
\[
\omega(a+b)+S(\omega\|(\rho^b)^a)=\omega(b)+S(\omega\|\rho^b)=S(\omega\|\rho),
\]
which implies the chain rule \eqref{eq:chin_rule}.

\end{proof}

\begin{example} Let $\Me=B(\mathcal H)$, the algebra of bounded operators on a Hilbert space
$\mathcal H$. The functional $\rho$ is represented as a density operator $\rho\in
B(\mathcal H)^+$ with finite trace, such that $\rho(a)=\mathrm{tr}[\rho a]$, ($\mathrm{tr}$ being the
usual trace on $B(\mathcal H)$). One can see that in this case,
\[
\rho^a= \exp(\log \rho +a).
\]

\end{example}

The following result is obtained from \cite[Prop. 5.3 and 5.4]{ekeland1999convex}.

\begin{lemma}\label{lemma:crho_gateaux} The function $C_\rho$ is Gateaux differentiable,
with Gateaux derivative at $b\in \Me^s$ given by  $C'_\rho(b)=\rho^b$. For $a,b\in \Me^s$, we have
\[
C_\rho(a)-C_\rho(b)\ge \rho^b(a-b).
\]

\end{lemma}

\subsection{The exponential Young function and its dual}\label{sec:young}

We now introduce  a conjugate pair of Young functions on the Banach spaces $\Me^s$ and
$(\Me^*)^s$.   Define
\begin{align*}
\Phi_\rho(a)&:=\frac12(C_\rho(a)+C_\rho(-a))-\rho(1),\qquad a\in \Me^s.\\
\Psi_\rho(\psi)&:=\frac12 \inf_{\substack{\omega_\pm\in (\Me^*)^+ \\ 2\psi=\omega_+-\omega_-}}
\left[F_\rho(\omega_+)+F_\rho(\omega_-)\right]+\rho(1),\qquad \psi\in
(\Me^*)^s.
\end{align*}

It was proved in \cite{jencova2006aconstruction} that $\Phi_\rho$ is a strictly convex and
continuous Young function $\Me^s\to \mathbb R$. We now look at the properties of
$\Psi_\rho$.

\begin{lemma}\label{lemma:dual_young} $\Psi_\rho$ is a strictly convex and weak*-lower
semicontinuous function on $(\Me^*)^s$, with effective domain
\[
\mathrm{Dom}(\Psi_\rho)=\{\psi\in (\Me^*)^s,\
\Psi_\rho(\psi)<\infty\}=\Pe_\rho-\Pe_\rho\subseteq \Me_*^s.
\]
\end{lemma}

\begin{proof} It is quite clear that  $\Psi_\rho(\psi)$ is finite if and only
if $\psi=\omega_+-\omega_-$ for some $\omega_\pm\in \Pe_\rho$. Further,  strict convexity
of $F_\rho$ implies 
that $\Psi_\rho$ is strictly convex as well. For the last statement we have to show
that  for
any $c>0$, the set $\{\psi\in (\Me^*)^s,\ \Psi_\rho(\psi)\le c\}$ is weak*-closed. So assume that
$(\psi_i)$ is a net in $(\Me^*)^s$  such that $\Psi_\rho(\psi_i)\le c$ and let $\psi_i\to
\psi$ in the weak*-topology. 
For each $\varepsilon>0$ and for all $i$ there are some 
functionals $\psi_{i,\pm}^\varepsilon\in \Me_*^+$ such that
$2\psi_i=\psi_{i,+}^\varepsilon-\psi_{i,-}^\varepsilon$ and 
\[
\frac12[F_\rho(\psi_{i,+}^\varepsilon)+
F_\rho(\psi_{i,-}^\varepsilon)]+\rho(1)\le c+\varepsilon.
\]
Using Proposition \ref{prop:Frho_properties} (ii), we obtain that
$\psi^\varepsilon_{i,\pm}\in \mathcal S_{K_\epsilon}$ with
$K_\epsilon=2(c+\varepsilon)-\rho(1)$.
 By Proposition 
\ref{prop:Frho_properties} (iii), $\Se_{K_\varepsilon}$ is weak*-compact, so that there
is a subnet $(\psi_j)$  and some $\psi^\varepsilon_\pm\in \mathcal M_*^+$ such that
$\psi^\varepsilon_{j,\pm}\to \psi^\varepsilon_\pm$. We therefore have
$\psi^\varepsilon_+-\psi^\varepsilon_-=\lim
\psi^\varepsilon_{j,+}-\psi^\varepsilon_{j,-}=2\psi$ and by weak*-lower
semicontinuity of $F_\rho$,
\[
\Psi_\rho(\psi)\le
\frac12[F_\rho(\psi^\varepsilon_+)+F_\rho(\psi^\varepsilon_-)]+\rho(1)\le \lim\inf_j
\frac12[F_\rho(\psi^\varepsilon_{j,+})+F_\rho(\psi^\varepsilon_{j,-})]+\rho(1)\le c+\epsilon.
\]
Since this holds for all $\varepsilon>0$, we have $\Psi_\rho(\psi)\le c$.

\end{proof}

\begin{prop}\label{prop:dual_young}
$\Psi_\rho$ is the Legendre-Fenchel conjugate of $\Phi_\rho$, with respect to the
dual pair $(\Me^s,(\Me^*)^s)$. In particular, $\Psi_\rho$ is a Young function on
$(\Me^s)^*$.
\end{prop}

\begin{proof} Since $F_\rho$ is weak*-lower semicontinuous, we see that
$C^*_\rho=F^{**}_\rho=F_\rho$. Let $D_\rho$ be given by  $D_\rho(a)=C_\rho(-a)$ for $a\in
\Me^s$, then $D_\rho^*(\psi)=C_\rho^*(-\psi)$ for $\psi\in (\Me^s)^*$. 
By  \cite[Cor. 2.3.5]{zalinescu2002convex}  and the fact that
$\Psi_\rho$ is weak*-lower semicontinuous, we obtain  $\Psi_\rho=\Phi_\rho^*$, so that
$\Psi_\rho$  is a Young  function on $(\Me^*)^s$ by  \cite[Lemma
3.4]{jencova2006aconstruction}.

\end{proof}

\subsection{The spaces $\eexp(\Me,\rho)$ and $\llog(\Me,\rho)$}

Using the Young functions ${\Phi_\rho}$ and ${\Psi_\rho}$, we construct the
corresponding  Banach spaces $B_{\Phi_\rho}$ and
$B_{\Psi_\rho}$ as in  Section \ref{sec:orlica}. 
The following is a consequence of the results of Section \ref{sec:young} and  \cite[Prop. 2]{jencova2006aconstruction}.

%

\begin{prop}\label{prop:embeddings1} We have $V_{\Phi_\rho}=\Me^s$ and
$B_{\Psi_\rho}=V_{\Psi_\rho}=\Pe_\rho-\Pe_\rho$. Moreover, $B_{\Psi_\rho}= B_{\Phi_\rho}^*$
(with equivalent norms) and we have the continuous embeddings
\[
\Me^s\sqsubseteq B_{\Phi_\rho},\qquad B_{\Psi_\rho}\sqsubseteq \Me_*^s.
\]

\end{prop}

Let us  now look at the  case when
$\Me$ is commutative. Since  $\rho$ is faithful, $\Me$  can be identified with the space 
$L_\infty(\Omega,\Sigma,\rho)$ where $\rho$ is a  finite measure on $(\Omega,\Sigma)$. Let
$\phi:\mathbb R\to \mathbb R$, $\phi(x)=\cosh(x)-1$ 
and let $\psi$ be its conjugate, then $\psi$ satisfies the $\Delta_2$ condition
$\psi(2u)\le K\psi(u)$ for some $K>0$. The exponential Orlicz space  $L^\phi(\Omega,\Sigma,\rho)$ is  the dual space of $L^\psi(\Omega,\Sigma,\rho)$.
Since the measure $\rho$ is finite, we have $L_\infty(X,\Sigma,\rho)\subseteq
L^\phi(X,\Sigma,\rho)$ and one can see 
 that the norm  obtained from our construction
coincides with  the Luxemburg-Nakano norm in $L^\phi(X,\Sigma,\rho)$. Hence 
$B_{\Phi_\rho}$
coincides with  the closure $E^\phi(X,\Sigma,\rho)$ of $L_\infty(X,\Sigma,\rho)$ in $L^\phi(X,\Sigma,\rho)$. We then have
\[
L^\psi(X,\Sigma,\rho)= E^\phi(X,\Sigma,\rho)^*= B_{\Psi_\rho}
\]
and $L^\phi(X,\Sigma,\rho)$ coincides with the second dual 
$B_{\Phi_\rho}^{**}$, see  \cite{rao1991theory} for details. These
facts were also pointed out in \cite{grasselli2010dualv1}.
It is therefore reasonable to identify the noncommutative counterpart of $L^\psi$ with the
space $B_{\Psi_\rho}$, while the noncommutative exponential Orlicz space should be
identified with $B_{\Phi_\rho}^{**}=B_{\Psi_\rho}^*$. Nevertheless, we will work with the
more tractable space $B_{\Phi_\rho}$, which is a strict subset of $B^{**}_{\Phi_\rho}$ in
general.

Let us denote $\eexp(\Me,\rho):= B_{\Phi_\rho}$, with the norm
$\|\cdot\|_{\exp,\rho}:=\|\cdot\|_{\Phi_\rho}$ and  $L_{\log}(\Me,\rho):=B_{\Psi_\rho}$, with the norm
$\|\cdot\|_{\log,\rho}:=\|\cdot\|_{\Psi_\rho}$. In the rest of this section, we will
identify $\Me_*$ with $L_1(\Me)$, so that $\Me_*^s$ is identified with
the space $L_1(\Me)^s$ of self-adjoint elements and $\Me_*^+$ with the cone
$L_1(\Me)^+$ of positive elements in $L_1(\Me)$. 

\begin{theorem}\label{thm:Llog_properties}
\begin{enumerate}
\item[(i)] $L_{\log}(\Me,\rho)=\Pe_\rho-\Pe_\rho$ is dense in  $L_1(\Me)^s$ and we have 
\[
L_p(\Me,\rho)^s\sqsubseteq L_{\log}(\Me,\rho)\sqsubseteq L_1(\Me)^s, \qquad 1<p\le \infty.
\]

\item[(ii)] $L_{\log}(\Me,\rho)^+:= L_{\log}(\Me,\rho)\cap L_1(\Me)^+=\Pe_\rho$ is  a
closed convex cone in $L_{\log}(\Me,\rho)$.

\item[(iii)] Let $\psi\in L_{\log}(\Me,\rho)$. Then $\|\psi\|_{\log,\rho}\le 1$ if and
only if there are some $\omega_\pm \in \mathcal P_\rho$ such that
$\psi=\frac12(\omega_+-\omega_-)$ and 
\[
F_\rho(\omega_+)+F_\rho(\omega_-)\le 2-2\rho(1).
\]

\end{enumerate}

\end{theorem}

\begin{proof}  By Proposition
\ref{prop:Frho_properties} (iv),  we see that $\mathrm{Dom}(\Psi_\rho)=\Pe_\rho-\Pe_\rho$ contains the
self-adjoint part $L_p(\Me,\rho)^s=L_p(\Me,\rho)^+-L_p(\Me,\rho)^+$ of $L_p(\Me,\rho)$, for any $p>1$, so $\llog(\Me,\rho)$ is dense in $L_1(\Me)^s$.
This also shows that $\Psi_\rho$ is finite valued on $L_p(\Me,\rho)^s$.
Since $L_p(\Me,\rho)^s\sqsubseteq L_1(\Me)^s$  and
$\Psi_\rho$ is lower semicontinuous on $L_1(\Me)^s$,  the restriction of $\Psi_\rho$ defines a convex and lower semicontinuous 
Young function  $L_p(\Me,\rho)^s\to \mathbb R$, which then must be continuous  by \cite[Cor. 2.5]{ekeland1999convex}. 
Let $B$ be the corresponding Banach
space, then by construction, $B$ is a closed subspace in  $L_{\log}(\Me,\rho)$ and using again \cite[Prop. 2]{jencova2006aconstruction}, we have 
\[
L_p(\Me,\rho)^s\sqsubseteq B\sqsubseteq \llog(\Me,\rho).
\]

Let now
$\omega\in L_{\log}(\Me,\rho)^+$, then  there are some $\omega_\pm
\in \mathcal P_\rho$ such that $2\omega=\omega_+-\omega_-$. It follows that
$2\omega+\omega_-=\omega_+\in \mathcal P_\rho$. By Proposition \ref{prop:Frho_properties}
(iv), this implies that we must have  $\omega\in \mathcal P_\rho$ as well. The fact that
the cone is closed in $L_{\log}(\Me,\rho)$ follows by the continuous embedding in
$L_1(\Me)^s$.

Assume that $\|\psi\|_{\log,\rho}\le 1$, equivalently, $\Psi_\rho(\psi)\le
1$. Then for any $n\in \mathbb N$, there are some $\omega_{\pm,n}\in \mathcal P_\rho$ such
that $\psi=\frac12(\omega_{+,n}-\omega_{-,n})$ and
$F_\rho(\omega_{+,n})+F_{\rho}(\omega_{-,n})\le 2(1+1/n-\rho(1))$. It then follows that
$\omega_{\pm,n}\in \mathcal S_C$ for some $C$ and all $n$. By Proposition
\ref{prop:Frho_properties} (iii), there is some subsequence such that $\omega_{\pm,
n_k}\to \omega_{\pm}$ in the $\sigma(\Me_*,\Me)$-topology. It follows that 
$\psi=\frac12(\omega_+-\omega_-)$ and by lower semicontinuity,
\[
F_{\rho}(\omega_+)+F_\rho(\omega_-)\le \lim\inf
F_{\rho}(\omega_{+,n_k})+F_\rho(\omega_{-,n_k})\le 2-2\rho(1).
\]
The converse is obvious.

\end{proof}

Let us now recall the embedding $i_{\infty,\rho}:\Me^s\to L_1(\Me)^s$, given by
\eqref{eq:infty_embedding}. Note that  $i_{\infty,\rho}(\Me^s)=L_\infty(\Me,\rho)^s\sqsubseteq
L_p(\Me,\rho)^s$, for all $1\le p\le \infty$.

\begin{theorem}\label{thm:embedding_exp} For any $1\le p<\infty$, $i_{\infty,\rho}$ extends to
a continuous embedding
\[
i_{\exp,\rho}: \eexp(\Me,\rho)\to L_p(\Me,\rho)
\]
and  $i_{\exp,\rho}(\eexp(\Me,\rho))$ is dense in $L_p(\Me,\rho)^s$.
\end{theorem}

\begin{proof} Let $a\in\Me^s$, $1\le p<\infty$ and let $1/p+1/q=1$.  By Theorem
\ref{thm:Llog_properties}, we have $L_q(\Me,\rho)^s\sqsubseteq L_{\log}(\Me,\rho)$. It
follows that for any $k\in L_q(\Me,\rho)$, we have
\[
\<i_{\infty,\rho}(a),k\>=\Tr[ak]\le \|a\|_{\exp,\rho}\|k\|_{\log,\rho}.
\]
Since  $\|k\|_{\log,\rho}\le M\|k\|_{q,\rho}$ for some
$M>0$, this shows that $i_{\infty,\rho}:\Me^s\to L_p(\Me,\rho)^s$ is continuous with respect to the norm
$\|\cdot\|_{\exp,\rho}$ in $\Me^s$ and therefore has a unique continuous extension
$i_{\exp,\rho}$. The rest follows from the fact that
$i_{\infty,\rho}(\Me^s)=L_\infty(\Me,\rho)^s$ is dense in $L_p(\Me,\rho)^s$ for any $p$.

\end{proof}

To summarize, we have for $1<p<\infty$: 
\begin{equation}\label{eq:embeddings}
L_\infty(\Me,\rho)\sqsubseteq E_{\exp}(\Me,\rho)\sqsubseteq L_p(\Me,\rho)\sqsubseteq L_{\log}(\Me,\rho)\sqsubseteq
L_1(\Me).
\end{equation}
Note that we have analogous properties for the classical exponential Orlicz spaces,
\cite[Prop. 8]{pistone1999theexponential}.

%
%
%
%
%


\begin{prop}\label{prop:monotone} Let $T:L_1(\Me) \to L_1(\Ne)$ be a channel.
The restriction of $T$ to $L_{\log}(\Me,\rho)$ is a contraction $L_{\log}(\Me,\rho)\to L_{\log}(\Ne,T(\rho))$.
The adjoint  $T^*$ extends uniquely to a contraction
$E_{\exp}(\Ne,T(\rho))\to E_{\exp}(\Me,\rho)$. 
\end{prop}

\begin{proof} By Proposition \ref{prop:monotone_F}, we have  $F_{T(\rho)}(T(\omega))\le F_\rho(\omega)$.
Let $\psi\in L_{\log}(\Me,\rho)$, then   
\begin{align*}
\Psi_{T(\rho)}(T(\psi))&\le \frac12 \inf_{2\psi=\psi_+-\psi_-}
(F_{T(\rho)}(T(\psi_+))+F_{T(\rho)}(T(\psi_-)))+T(\rho)(1)\\
&\le \frac12 \inf_{2\psi=\psi_+-\psi_-}
(F_{\rho}(\psi_+)+F_{\rho}(\psi_-))+\rho(1)=\Psi_\rho(\psi).
\end{align*}
By definition of  $L_{\log}(\Me,\rho)$ and duality, this implies the statement.

\end{proof}

\section{The quantum exponential manifold}

Let $\tilde \Fe$ denote the set of all faithful functionals in the cone $\Me_*^+$. In this
section we will construct a Banach manifold structure on $\tilde \Fe$, using an extension
of Theorem \ref{thm:conjugate} to perturbations in $E_{\exp}(\Me,\rho)$.

\subsection{Extended perturbations}

Since the effective domain of the relative entropy function $F_\rho$ is the positive cone $\mathcal P_\rho=\llog(\Me,\rho)^+$, it can be
regarded as a strictly convex function on $\llog(\Me,\rho)$. In this section, we will
investigate the function $F_\rho$  and  its conjugate with respect to the dual pair
$(\llog(\Me,\rho),\eexp(\Me,\rho))$.

 We first note  that $F_\rho$ as a function on $\llog(\Me,\rho)$ is weak*-lower
 semicontinuous. Indeed, since $\Me^s$ is norm dense in $\eexp(\Me,\rho)$, the
weak*-topology on $L_{\log}(\Me,\rho)$ coincides with the restriction of the $\sigma(\Me_*^s,\Me^s)$-topology on
norm-bounded subsets. By Proposition \ref{prop:Frho_properties} (iii), the claim will follow by
the next   Lemma.


\begin{lemma}\label{lemma:fhro_on_llog} For each $C\in \mathbb R$, 
 $\mathcal S_C$ is norm-bounded in $L_{\log}(\Me,\rho)$.
\end{lemma}

\begin{proof} We may assume that $C\ge -\rho(1)$, otherwise $\mathcal S_C$ is empty.
If $\omega \in \mathcal S_C$, then $\omega\in \Me_*^+$ and we have (using the
decomposition $2(\tfrac12\omega)=\omega-0$)
\[
\Psi_\rho(\frac12 \omega)\le \frac12F_\rho(\omega)+\rho(1)\le \frac12 C+\rho(1).
\]
If $\Psi_\rho(\frac12 \omega)\le 1$, then $\|\omega\|_{\log,\rho}\le 2$, otherwise we have
by \cite[Lemma 3.3]{jencova2006aconstruction} that $\|\frac12\omega\|_{\log,\rho}\le
\Psi_\rho(\frac12\omega)\le \frac12 C+\rho(1)$. Hence $\|\omega\|_{\log,\rho}\le
\max\{2,C+2\rho(1)\}$. 

\end{proof}

Let us now recall the Legendre-Fenchel conjugate function $C_\rho$, defined in  \eqref{eq:frho_conjugate}.  It is easily seen that $C_\rho$ is bounded
over the unit ball with respect to $\|\cdot\|_{\exp,\rho}$ in $\Me^s$. By \cite[Cor.
2.4]{ekeland1999convex}, this implies that $C_\rho$ is continuous (in fact,
locally Lipschitz) with respect to this  norm.  It follows that  $C_\rho$ extends uniquely to a continuous function
$\eexp(\Me,\rho)\to \mathbb R$, which will be again denoted by $C_\rho$. The next result
shows that this extension is the conjugate function of $F_\rho$ with respect to the dual
pair 
$(\llog(\Me,\rho),\eexp(\Me,\rho))$. 

\begin{theorem}\label{thm:extended_conjugate} For $h\in \eexp(\Me,\rho)$, we have
\[
C_\rho(h)=\sup_{\omega\in L_{\log}(\Me,\rho)} \omega(h)-F_\rho(\omega).
\]
The supremum is attained at a unique functional $\rho^h\in \Pe_\rho$. Moreover, $\rho^h$
is faithful, 
$C_\rho(h)=\rho^h(1)$ and the map
$\eexp(\Me,\rho)\ni h\mapsto
\rho^h\in \Me_*$ is norm-to-norm continuous. 

\end{theorem}

\begin{proof}
Let $a_n\in \Me^s$ be a sequence such that
$\|h-a_n\|_{\rho}\to 0$. By putting $a=2a_n$ and $b=a_n$ in Lemma
\ref{lemma:crho_gateaux}, we obtain the inequality
\[
C_\rho(2a_n)-C_\rho(a_n)\ge \rho^{a_n}(a_n),\qquad \forall n.
\]
By continuity of $C_\rho$, this implies that $\{\rho^{a_n}(a_n)\}_n$ is a  bounded sequence, so that
also 
\[
\{F_{\rho}(\rho^{a_n})=\rho^{a_n}(a_n)-C_\rho(a_n)\}_n
\]
is bounded and therefore  $\rho^{a_n}\in \Se_K$ for some $K$. By Proposition
\ref{prop:Frho_properties} (iii) we may assume (by restricting to a subsequence) that there
is some $\sigma\in \Se_K$ such that $\rho^{a_n}\to \sigma$ in the
$\sigma(\Me_*,\Me)$-topology. Since $\Se_K$ is norm
bounded in $L_{\log}(\Me,\rho)$ (Lemma \ref{lemma:fhro_on_llog}) and the weak*-topology
coincides with the $\sigma(\Me_*,\Me)$-topology on $\Se_K$, it can be seen that
$\rho^{a_n}(a_n)\to \sigma(h)$.
For $\omega\in \Pe_\rho$, we get by definition of $\rho^{a_n}$ and lower semicontinuity
\begin{align*}
\omega(h)-F_\rho(\omega)&=\lim_n (\omega(a_n)-F_\rho(\omega))\le \lim_n C_\rho(a_n)=\lim_n
(\rho^{a_n}(a_n)-F_\rho(\rho^{a_n}))\\&\le \sigma(h)-F_\rho(\sigma).
\end{align*}
It follows that $\sigma$ is a maximizer of $\omega(h)-F_\rho(\omega)$ and by strict convexity
of $F_\rho$ such  maximizer is unique. Let us denote $\rho^h:=\sigma$. Note that
we have $C_\rho(h)=\lim_n C_\rho(a_n)=\lim_n \rho^{a_n}(1)=\rho^h(1)$ and the above
computation also implies that 
\[
C_\rho(h)\le \rho^h(h)-F_\rho(\rho^h)=\sup_{\omega\in \Pe_\rho} \omega(h)-F_\rho(\omega).
\]
  On the other hand, we obtain using
\eqref{eq:perturbed_entropy} and lower semicontinuity of $S$:
\begin{equation}\label{eq:preturbed_entropy_h_onesided}
S(\omega\|\rho)=\lim_n (\omega(a_n)+S(\omega\|\rho^{a_n}))\ge \omega(h)+S(\omega\|\rho^h),
\quad \forall \omega\in \mathcal P_\rho.
\end{equation}
Putting $\omega=\rho$, we see that $S(\rho\|\rho^h)$ is finite, so that $\rho^h$ must be
faithful.
Further, putting  $\omega=\rho^h$ it follows that $\rho^h(h)-F_\rho(\rho^h)=\rho^h(1)=C_\rho(h)$.
We also get $\lim_nS(\rho^h\|\rho^{a_n})=0$, so that $\rho^{a_n}\to \rho^h$ strongly in
$\Me_*$, this is easily extended to all sequences $h_n\to h$ in $\eexp(\Me,\rho)$. 

\end{proof}

We now extend the equalities in Theorem \ref{thm:conjugate} to all elements in
$\eexp(\Me,\rho)$. 

\begin{lemma}\label{lemma:chainrule_onehalf} Let $h\in E_{\exp}(\Me,\rho)$.
\begin{enumerate}
\item[(i)] For $a\in \Me^s$, we have the chain rule
\[
\rho^{h+a}=(\rho^h)^a,\qquad C_\rho(h+a)=C_{\rho^h}(a).
\]
\item[(ii)] The norm $\|\cdot\|_{\exp,\rho^h}$ is continuous with respect to
$\|\cdot\|_{\exp,\rho}$ on $\Me^s$.

\end{enumerate}

\end{lemma}

\begin{proof}
 Let $a_n\in\Me^s$ be a sequence such that $\|h-a_n\|_{\exp,\rho}\to 0$. 
For $a\in \Me^s$ we have $C_\rho(a_n+a)\to C_\rho(h+a)$ and 
$\rho^{a_n+a}\to \rho^{h+a}$ strongly, by Theorem \ref{thm:extended_conjugate}. 
Since also $\rho^{a_n}\to \rho^h$ strongly, we have $(\rho^{a_n})^a\to (\rho^h)^a$
strongly, by \cite[Thm. 1.1]{donald1991continuity}. By the chain rule \eqref{eq:chin_rule}, we obtain
\[
\rho^{h+a}=(\rho^h)^a,\qquad C_\rho(h+a)=\rho^{h+a}(1)=(\rho^h)^a(1)=C_{\rho^h}(a). 
\]
To prove (ii),  note that (i)  implies 
\[
\Phi_{\rho^h}(a)=\frac12(C_\rho(h+a)-C_{\rho}(h)+C_\rho(h-a)-C_\rho(h)),\qquad a\in \Me^s.
\]
By continuity of $C_\rho$, this shows that there is 
some $\delta>0$ such that $\Phi_{\rho^h}(a)<1$ whenever $\|a\|_{\exp,\rho}<\delta$, this
proves (ii). 

\end{proof}

\begin{theorem}\label{thm:extended_conjugate_equalities}
Let $h\in \eexp(\Me,\rho)$. Then 
\begin{equation}\label{eq:perturbed_entropy_extended}
\omega(h)+S(\omega\|\rho^h)=S(\omega\|\rho),\qquad \omega \in \Pe_\rho.
\end{equation}
Moreover, $\eexp(\Me,\rho)=\eexp(\Me,\rho^h)$ (equivalent norms) and we have the  chain rule
\begin{equation}\label{eq:chin_rule_extended}
\rho^{h+k}=(\rho^h)^k,\qquad C_\rho(h+k)=C_{\rho^h}(k),\qquad h,k\in \eexp(\Me,\rho).
\end{equation}

\end{theorem}

\begin{proof} Let $a_n\in \Me^s$, $\|a_n-h\|_{\exp,\rho}\to 0$. By Lemma
\ref{lemma:chainrule_onehalf} and Theorem \ref{thm:extended_conjugate}, we obtain that
also $h\in \eexp(\Me,\rho^h)$ and
$\|a_n-h\|_{\exp,\rho^h}\to 0$.  Moreover, 
\[
(\rho^h)^{-h}=\lim_n (\rho^h)^{-a_n}=\lim_n \rho^{h-a_n}=\rho^0=\rho.
\]
Replacing $\rho$ by $\rho^h$ and $h$ by $-h$ in \eqref{eq:preturbed_entropy_h_onesided},
we obtain
\[
S(\omega\|\rho^h)\ge -\omega(h)+S(\omega\|\rho).
\]
Together with \eqref{eq:preturbed_entropy_h_onesided}, this implies
\eqref{eq:perturbed_entropy_extended}. Similarly, using  this replacement in Lemma
\ref{lemma:chainrule_onehalf} (ii), we obtain that $\eexp(\Me,\rho)=\eexp(\Me,\rho^h)$ with
equivalent norms. The chain rule \eqref{eq:chin_rule_extended} is now proved from
\eqref{eq:perturbed_entropy_extended} exactly as in the proof of Theorem
\ref{thm:conjugate}.

\end{proof}

\begin{coro}\label{coro:conjugate_explog}
With respect to the dual pair $(L_{\log}(\Me,\rho), \eexp(\Me,\rho))$, we have
$C_\rho=F_\rho^*$ and $F_\rho=C_\rho^*$. Moreover, $C_\rho$ is strictly convex and Gateaux differentiable on
$\eexp(\Me,\rho)$,
with the Gateaux derivative $C'_\rho(h)=\rho^h$, and  $h\mapsto \rho^h$ defines an  injective
and norm-to-weak*-continuous  map $\eexp(\Me,\rho)\to L_{\log}(\Me,\rho)$.

\end{coro}

\begin{proof} The first part is clear from Theorem \ref{thm:extended_conjugate} and
weak*-lower semicontinuity of $F_\rho$. Differentiability of $C_\rho$ is then obtained from
\cite[Prop. 5.3]{ekeland1999convex}. Injectivity of the map $h\mapsto \rho^h$ follows by Theorem
\ref{thm:extended_conjugate_equalities}, this also implies strict convexity of $C_\rho$
(e.g. as in the proof of \cite[Thm. 7.3]{jencova2006aconstruction}). For continuity,  see e.g.
\cite{zalinescu2002convex}. 

\end{proof}

\subsection{Exponential families in $\Me_*^+$}

 Let $E\subseteq E_{\exp}(\Me,\rho)$ be a
closed subspace. The set
\[
\mathcal E_\rho(E):=\{\rho^h,\ h\in E\}
\]
will be called an exponential family (at $\rho$). The set  $\mathcal E_\rho:=\mathcal
E_\rho(E_{\exp}(\Me,\rho))$ will be called the full exponential family (at $\rho$).
For the following characterization of  elements of $\mathcal E_\rho$, note that by
\eqref{eq:donald} 
\[
\omega\mapsto S(\omega\|\rho)-S(\omega\|\sigma)
\]
defines an affine map $h_{\sigma,\rho}: \mathcal P_\rho \to [-\infty,\infty)$ such that
$h_{\sigma,\rho}(0)=0$.  

\begin{coro}\label{coro:rhoh}
Let $\sigma\in \Me_*^+$. Then $\sigma=\rho^h$ for some $h\in E_{\exp}(\Me,\rho)$ if and
only if there is some $C>-\rho(1)$ such that $h_{\sigma,\rho}$ is bounded and $\sigma(\Me_*,\Me)$-continuous on the set
$\mathcal S_C$. In this case $h$ coincides with $h_{\sigma,\rho}$
on $\mathcal P_\rho$. 

\end{coro}

\begin{proof} Assume that $\sigma=\rho^h$ for $h\in E_{\exp}(\Me,\rho)$, then by Theorem
\ref{thm:extended_conjugate_equalities}, we see that $h(\omega)=h_{\sigma,\rho}(\omega)$ for $\omega\in
\mathcal P_\rho$. Since   the $\sigma(\Me_*,\Me)$-topology coincides with the weak*-topology on $\mathcal S_C$,
 the assertion follows from $L_{\log}(\Me,\rho)=E_{\exp}(\Me,\rho)^*$.

 Assume conversely that $h_{\sigma,\rho}$ has the stated properties on $\mathcal S_C$ for
 some $C>-\rho(1)$. Then the same  is true for any $C'\in \mathbb R$, since for $C'>C$, there is some $t\in
 [0,1]$ such that $F_\rho(t\omega+(1-t)\rho)\le tC'-(1-t)\rho(1)\le C$ for any $\omega\in
 \mathcal S_{C'}$.

Now note that by Theorem \ref{thm:Llog_properties} (iii) and Proposition
\ref{prop:Frho_properties} (ii), the unit ball in $L_{\log}(\Me,\rho)$ is
a subset of $\mathcal S_C-\mathcal S_C$ for $C=2\rho(1)-1$, so that $h_{\sigma,\rho}$
extends to a bounded linear map on $L_{\log}(\Me,\rho)$, moreover, since the
weak*-topology coincides with the $\sigma(\Me_*,\Me)$-topology on bounded subsets in
$L_{\log}(\Me,\rho)$, this extension is weak*-continuous and hence defines an element
$h\in E_{\exp}(\Me,\rho)$. For $\omega\in \Pe_\rho$,  we get 
\[
\omega(h)-F_\rho(\omega)=-F_\sigma(\omega)\le
\sigma(1)=\sigma(h)-F_\rho(\sigma),
\]
so that $\sigma=\rho^h$.

\end{proof}

We are now ready to introduce a Banach manifold structure on 
$\tilde{\mathcal F}$
 using the parametrization $h\mapsto \rho^h$, similarly as in \cite{jencova2006aconstruction} for the
set  of  faithful states. For $\rho\in \tilde \Fe$, let
$V_\rho$ be the open unit ball in $E_{\exp}(\Me,\rho)$ and $s_\rho: V_\rho\to\tilde
\Fe$ the map $h\mapsto \rho^h$.  We construct a $C^\infty$-atlas on $\tilde\Fe$ as 
\[
\{(U_\rho, e_\rho), \ \rho\in \tilde\Fe\}
\]
where $U_\rho=s_\rho(V_\rho)$ and $e_\rho=s_\rho^{-1}\vert_{U_\rho}$. 
To show that this is indeed a $C^\infty$-atlas, it is enough to notice that 
if $U_{\rho_1}\cap U_{\rho_2}\ne \emptyset$, then we must have $\rho_1=\rho_2^k$ for some 
 $k=E_{\exp}(\Me,\rho_2)$, and 
\[
 e_{\phi_1}(U_{\rho_1}\cap U_{\rho_2})=\{h_1\in E_{\exp}(\Me,\rho_1),\
 \|h_1\|_{\exp,\rho_1}<1,\ \|h_1+k\|_{\exp,\rho_2}<1\}.
 \]
The proof is finished  similarly as in \cite{jencova2006aconstruction}, using the
equivalence of the two norms $\|\cdot\|_{\exp,\rho_1}$ and $\|\cdot\|_{\exp,\rho_2}$.
It is also clear that the connected components of $\tilde\Fe$ are exactly the full
exponential families $\mathcal E_\rho$, $\rho\in \tilde\Fe$.

\subsection{The canonical divergence}

Using  Corollary \ref{coro:conjugate_explog}, we can
introduce a canonical divergence in the connected component $\mathcal E_\rho$, $\rho\in
\tilde\Fe$, as the
Bregman divergence associated with $C_\rho$:
\[
D_\rho(h\|k):= C_\rho(h)-C_\rho(k)-\<C'_\rho(k),h-k\>,\qquad h,k\in \eexp(\Me,\rho).
\]
\begin{theorem}\label{thm:canonical} Let $\rho\in \tilde \Fe$, $h,k\in \eexp(\Me,\rho)$.
Then
\begin{enumerate}
\item[(i)] $D_\rho(h\|k)=S(\rho^k\|\rho^h)-(\rho^k-\rho^h)(1)$.

\item[(ii)]  $D_\rho(h\|k)\ge 0$, with equality if and only if $h=k$.
\item[(iii)] The function $D_\rho: E_{\exp}(\Me,\rho)\times E_{\exp}(\Me,\rho)\to \mathbb R$ is jointly 
continuous, and it is  strictly convex and Gateaux differentiable in the first variable.
\item[(iv)] $D_\rho$ satisfies the generalized Pythagorean relation
\[
D(h\|k)+D(k\|l)=D(h\|l)+(\rho^k-\rho^l)(k-h),\quad h,k,l\in\eexp(\Me,\rho).
\]
\end{enumerate}

\end{theorem}

\begin{proof} The statement (i) is obtained  from Theorem \ref{thm:extended_conjugate_equalities}. 
The rest follows by the properties of the Bregman divergence. More explicitly, (ii) can be
seen from  \cite[Prop. 5.4]{ekeland1999convex} and strict convexity of $C_\rho$. To prove joint continuity, let $h_n$ and $k_n$ be two sequences such that $h_n\to h$, $k_n\to k$ in $E_{\exp}(\Me,\rho)$. By Corollary \ref{coro:conjugate_explog}, we have $C'_\rho(k)=\rho^k$. 
 Since $k\mapsto \rho^k$ is norm to weak*-continuous, it follows that
$\rho^{k_n}$ is a norm-bounded sequence in $L_{\log}(\Me,\rho)$, this implies that
$\rho^{k_n}(h_n-k_n)\to \rho^k(h-k)$, we then have $D_\rho(\rho^{k_n}\|\rho^{h_n})\to
D_\rho(\rho^h\|\rho^k)$ by continuity
of $C_\rho$. The rest of (iii) is straightforward from properties of $C_\rho$.  The
Pythagorean relation (iv) is clear  from the definition.

\end{proof}

\subsection{Sufficient channels and invariance}

Let $\mathcal E$ be a subset of nonzero elements in $\Me_*^+$ and let $T: L_1(\Me)\to
L_1(\Ne)$ be a channel.  We say that $T$ is sufficient for
$\mathcal E$ if there is a channel $S:L_1(\Ne)\to L_1(\Me)$ such that 
\[
S\circ T(\sigma)=\sigma,\qquad \forall \sigma\in \mathcal E.
\]
In this situation, $S$ will be called a recovery channel for $T$ on $\mathcal E$.

The notion of a sufficient channel was introduced by Petz \cite{petz1986sufficient,petz1988sufficiency} in the situation
when $\mathcal E$ is a set of states. Since the channels are trace preserving, the
extension to positive functionals is straightforward.

\begin{theorem}[\cite{petz1986sufficient,petz1988sufficiency}]\label{thm:suff} Let $T$ be a 2-positive channel and
assume that there is some faithful element $\rho\in \mathcal E$ such that $\mathcal
E\subseteq \Pe_\rho$. The
following are equivalent.
\begin{enumerate}
\item[(i)] $T$ is sufficient for $\mathcal E$;
\item[(ii)] $S(T(\sigma)\|T(\rho))=S(\sigma\|\rho)$;
\item[(iii)] The Petz dual $T_\rho$ of $T$ with respect to $\rho$ is a recovery channel
for $\mathcal E$.

\end{enumerate}
\end{theorem}

We will study the case when  $\mathcal E=\mathcal E_\rho(E)$ is
an exponential family at some $\rho\in \tilde \Fe$. Then the conditions of the above theorem
are fulfilled.

\begin{theorem}\label{thm:sufficient} Let $\rho \in \tilde \Fe$, $h\in E_{\exp}(\Me,\rho)$.
Let $T:L_1(\Me)\to L_1(\Ne)$ be a 2-positive channel and let $T_\rho$ be the Petz dual of $T$ with
respect to $\rho$. The following are equivalent.
\begin{enumerate}
\item[(i)] $T$ is sufficient with respect to $\{\rho,\rho^h\}$.
\item[(ii)] $T(\rho^h)=T(\rho)^{h_0}$ for some $h_0\in E_{\exp}(\Ne,T(\rho))$ and
$h=T^*(h_0)$
\item[(iii)] $T^*\circ T_\rho^*(h)=h$.

\end{enumerate}

\end{theorem}

\begin{proof} Since $T_\rho^*$ defines a map $E_{\exp}(\Me,\rho)\to E_{\exp}(\Ne,T(\rho))$,
we have for $\omega_0\in \mathcal P_{T(\rho)}$,
\begin{align}
\omega_0(T_\rho^*(h))-F_{T(\rho)}(\omega_0)&=
F_\rho(T_\rho(\omega_0))-F_{\rho^h}(T_\rho(\omega_0))-F_{T(\rho)}(\omega_0)\notag \\
&\le -F_{\rho^h}(T_\rho(\omega_0))\le \rho^h(1)=C_\rho(h).\label{eq:suf2}
\end{align}
Here we have used Theorem \ref{thm:extended_conjugate_equalities},  monotonicity of relative entropy together with the fact that
$\rho=T_\rho\circ T(\rho)$, and Proposition \ref{prop:Frho_properties} (ii).
Assume (i), then by Theorem
\ref{thm:sufficient} (ii) and (iii) we get
\[
T(\rho^h)(T_\rho^*(h))-F_{T(\rho)}(T(\rho^h))= \rho^h(h)-F_{\rho}(\rho^h)=C_\rho(h).
\]
It follows that the maximum in  \eqref{eq:suf2} is attained at $\omega_0=T(\rho^h)$, so
that  $T(\rho^h)=T(\rho)^{h_0}$, with
$h_0=T_\rho^*(h)$. By a similar computation, we obtain
\[
\omega(T^*(h_0))-F_\rho(\omega)\le C_{T(\rho)}(h_0),\qquad \forall \omega \in \mathcal P_\rho
\]
and equality is attained for $\omega=\rho^h$. Hence  $\rho^h=\rho^{T^*(h_0)}$, so that
$h=T^*(h_0)=T^*(T_\rho^*(h))$ by injectivity of the map $h\mapsto \rho^h$. This proves
that (i) implies both (ii) and (iii).
Assume (ii), then we have 
\begin{align*}
C_\rho(h)&=\rho^h(1)=T(\rho^h)(1)=T(\rho)^{h_0}(1)=C_{T(\rho)}(h_0)
=T(\rho^h)(h_0)-F_{T(\rho)}(T(\rho^h))\\ &=\rho^h(h)-F_{T(\rho)}(T(\rho^h))
\ge \rho^h(h)-F_\rho(\rho^h)=C_\rho(h).
\end{align*}
This implies (i) by Theorem \ref{thm:sufficient} (ii).  Finally, from (iii) and  $T_\rho\circ T(\rho)=\rho$, we have
\begin{align*}
C_\rho(h)&\ge \sup_{\omega_0\in \mathcal P_{T(\rho)}} 
T_\rho(\omega_0)(h)-F_\rho(T_\rho(\omega_0))\ge \sup_{\omega_0}
\omega_0(T_\rho^*(h))-F_{T(\rho)}(\omega_0)\\
&\ge T(\rho^h)(T_\rho^*(h))-F_{T(\rho)}(T(\rho^h))\ge \rho^h(T^*\circ T_\rho^*(h))
-F_\rho(\rho^h)=C_\rho(h).
\end{align*}
This shows that $F_\rho(\rho^h)=F_{T(\rho)}(T(\rho^h))$, which implies (i) by Theorem
\ref{thm:sufficient}.

\end{proof}

\begin{coro}\label{coro:suf}
Let $\rho\in \tilde\Fe$ and let $\mathcal E=\{\rho^h,\ h\in E_0\}$ for some subset
$E_0\subseteq E_{\exp}(\Me,\rho)$. Let $E\subseteq E_{\exp}(\Me,\rho)$ be the closed linear span of $E_0$. 
Let $T:L_1(\Me)\to L_1(\Ne)$ be a 2-positive channel sufficient with respect
to $\mathcal E$. Then 
\begin{enumerate}
\item[(i)] $T$ is sufficient for the exponential family $\mathcal E_\rho(E)$.
\item[(ii)] $T_\rho^*\vert_E$ is an isometric isomorphism of $E$
onto $T_\rho^*(E)$ and we have
$T(\mathcal E_\rho(E))=\mathcal E_{T(\rho)}(T_\rho^*(E))$, and
$T(\rho^h)=T(\rho)^{T^*_\rho(h)}$, for $h\in E$.
\end{enumerate}

\end{coro}

\section{Conclusions}

We have constructed an exponential manifold structure over the set $\tilde \Fe$ of
faithful positive functionals on a von Neumann algebra, which in the commutative case
coincides with a restriction of the Pistone-Sempi construction. The manifold is based on
the Araki relative entropy  and its conjugate $C_\rho$, playing the role of the moment generating function
from the classical theory. We showed the relation of the obtained structures to Kosaki
$L_p$ spaces and proved an invariance property of the exponential manifold.  Note that the
function $C_\rho$ was only proved to be Gateaux
differentiable, so we do not get the full power
of the Pistone-Sempi construction. Nevertheless, the manifold admits a canonical
divergence satisfying a generalized Pythagorean relation.

%
%
%
\subsubsection*{Acknowledgements}

 The research was supported by the grants VEGA 1/0142/20 and  the Slovak Research and
Development Agency grant APVV-20-0069.


\end{document}